\documentclass[a4paper,11pt]{article}
\usepackage{pos}

\title{Moment Problems and Spectral Functions}

\usepackage{amsmath}
\usepackage{amsthm}
\usepackage[capitalize]{cleveref}

\newcommand{\R}{\mathbb{R}}
\renewcommand{\H}{\mathbb{H}}
\newcommand{\C}{\mathbb{C}}
\newcommand{\D}{\mathbb{D}}
\renewcommand{\P}{\mathbb{P}}

\let\Im\undefined

\DeclareMathOperator{\Im}{Im}

\newtheorem{theorem}{Theorem}
\newtheorem{lemma}{Lemma}

\author*[a]{Ryan Abbott}
\author[b]{William Jay}
\author[c]{Patrick Oare}

\affiliation[a]{Physics Department, Columbia University, New York, NY 10027, USA}

\affiliation[b]{Department of Physics, Colorado State University, Fort Collins, CO 80523, USA}

\affiliation[c]{Physics Department, Brookhaven National Laboratory, Upton, NY 11973, USA}

\emailAdd{rwa2110@columbia.edu}

\FullConference{The 42nd International Symposium on Lattice Field Theory (LATTICE2025)\\
2-8 November 2025\\
Tata Institute of Fundamental Research, Mumbai, India\\}

\abstract{
  Nevanlinna-Pick interpolation and moment problems use the analytic structures provided by causality in order to provide rigorous bounds on smeared spectral functions.
  This proceedings discusses Nevanlinna-Pick interpolation and moment problems and reviews some useful results,  including a simple proof that the space of causal data in Nevanlinna--Pick interpolation is convex.
}

\begin{document}
\maketitle

\section{Introduction}
Analytic continuation is ubiquitous in the physical sciences, especially in
lattice field theory. One instance of an analytic continuation problem
is that of reconstructing a spectral density from the Euclidean
lattice data, which requires numerically inverting the Laplace
transform
\begin{equation}
C_{E}(t) = \int dE \, \rho(E) e^{-E t},
\end{equation}
where $C_{E}(t)$ is a Euclidean-time lattice correlation function, and
$\rho(E)$ is its associated spectral density.
Several methods have been proposed to tackle this ill-posed inverse
problem, including the Backus-Gilbert~\cite{Backus:1968svk,Hansen:2017mnd} and
Hansen-Lupo-Tantalo~\cite{Hansen:2019idp} methods, the latter of which
has enabled several recent phenomenological applications~\cite{ExtendedTwistedMassCollaborationETMC:2022sta,Evangelista:2023fmt}, alongside many other proposed methods~\cite{Lawrence:2024hjm,Bailas:2020qmv,Bruno:2024fqc,Horak:2021syv,DelDebbio:2024lwm}.
One difficulty of these approaches is that both introduce a bias to make the problem tractable, resulting in hard-to-quantify systematic uncertainties.
Nevanlinna-Pick interpolation~\cite{Bergamaschi:2023xzx,PhysRevLett.126.056402} and moment
problems~\cite{Abbott:2025snz} have recently arisen as methods that
use the analytic structures imposed by causality in order to
derive rigorous bounds on smeared spectral functions, thus allowing for precisely quantified systematic uncertainties.
Similar methods based on Nevanlinna--Pick interpolation have also been recently been used with a somewhat different perspective in a condensed-matter context~\cite{Yu_2024,Huang:2022qsb,PhysRevLett.126.056402}.
This proceedings provides a review of both Nevanlinna-Pick interpolation and moment problems from a common perspective, as well as some further developments on applying these methods to lattice QFT correlators.
In particular, \cref{sec:review} reviews the methods, while
\cref{sec:geometry} provides a simple proof that the space of causal correlation functions is convex.

\section{Review of Analyticity Methods}\label{sec:review}
Both Nevanlinna--Pick (NP) interpolation and the Hamburger moment problem
are concerned with reconstructing information about a positive density
$\rho$, given some set of constraints on $\rho$ derived from
Euclidean-time data.
In the NP case, $\rho_{\text{NP}}$ is the usual spectral density, while in the
Hamburger case $\rho_{\text{Moment}}$ is the spectral density
associated to the transfer matrix; the two are formally related
by~\cite{Abbott:2025snz}
\begin{equation}
\label{eq:4}
\rho_{\text{NP}}(E) = \lambda \rho_{\text{Moment}}(\lambda)
\end{equation}
where $\lambda = e^{-aE}$ is the transfer-matrix eigenvalue associated
to the energy $E$. Throughout this work the two different spectral
densities $\rho_{\text{NP}}$ and $\rho_{\text{Moment}}$ will be
used interchangeably and denoted by $\rho$.
For simplicity, we consider only the Hamburger moment problem, which
corresponds to the case of a staggered-fermion, thermal correlation function (the Stieltjes and Hausdorff moment problems relevant for Wilson-like fermions are considered in Ref~\cite{Abbott:2025snz}).

Both methods provide information on $\rho$ via rigorous bounds on the
\emph{Stieltjes transform} $G(z)$, defined by
\begin{equation}
  \label{eq:stieltjes-def}
    G(z) = \int \frac{\rho(x) \, dx}{z - x + i 0}.
\end{equation}
Note that as a consequence of the positivity of $\rho$, the Stieltjes
transform of $\rho$ satisfies $\Im G(z) \geq 0$ whenever $\Im z \geq 0$, so $G$ can be viewed as a map $G : \H \to \H$, where $\H = \{z \in \C \mid \Im c \geq 0\}$ is the upper-half plane.
In the Nevanlinna-Pick interpolation, the Stieltjes transform is
physically interpreted as the analytic continuation of the correlation
function, and \cref{eq:stieltjes-def} is exactly the
K{\"a}ll{\'e}n-Lehmann spectral representation.
In either case
$G(z)$ can also be interpreted as a smeared spectral spectral function
due to the relation
\begin{equation}
\label{eq:imag-G-eq-cauchy}
  \frac{1}{\pi} \Im G(x_0 + i \epsilon)
  = \int dx \, \rho(x) K^{\text{Cauchy}}_{x_0, \epsilon}(x),
\end{equation}
which is valid for \emph{finite} $\epsilon$, and where $K^{\text{Cauchy}}_{x_0, \epsilon}(\omega)$ is the Cauchy
kernel defined by
\begin{equation}
  \label{eq:cauchy-kernel-defn}
  K^{\text{Cauchy}}_{x_0, \epsilon}(x)
  = \frac{1}{\pi} \Im \frac{1}{x - (x_0 + i \epsilon)}
  = \frac{1}{\pi} \frac{\epsilon}{(x - x_0)^2 + \epsilon^2}.
\end{equation}
As noted in Ref.~\cite{Bergamaschi:2023xzx}, this provides a
connection between analyticity-based methods and other approaches that
rely on reconstructing smeared spectral densities. 
In the limit $\epsilon \to 0$, the Cauchy kernel approaches $\delta(x
- x_0)$, resulting in the standard relation
\begin{equation}
  \label{eq:stieltjes-perron}
  \rho(x) = \lim_{\epsilon \to 0} \frac{1}{\pi} \Im G(x + i \epsilon)
\end{equation}
which is referred to as the Stieltjes--Perron inversion formula in
the mathematical literature.

The largest practical difference between Nevanlinna-Pick interpolation
and the Hamburger moment problem arises in how a given set of
Euclidean-time correlator data is transformed into the inputs of the
problem. Pick interpolation assumes that $G(z)$ has been evaluated at
a discrete set of (imaginary) Matsubara frequencies
$\{i\omega_{\ell}\}$ via the relation
\begin{equation}
G(i \omega_\ell) = \int_0^{\infty} C_E(t) e^{-i \omega_\ell t} \, dt
\end{equation}
where $C_E(t)$ is the Euclidean-time correlation function.
Note that from a literal interpretation of the preceding equation it seems this approach assumes the correlation function is known at all times
$t > 0$, whereas in most Monte-Carlo calculations the correlation
function is only known at discrete times $t \in \{0, \dots, N_t-1\}$.
One possibility is to use an interpolation to estimate the
energy-domain input data (see Ref.~\cite{Bergamaschi:2023xzx} for details).
In contrast, the Hamburger moment problem operates directly on the
Euclidean data by recognizing the correlation function as the moments of
the distribution:
\begin{equation}
C_E(t) \equiv C_t = \int d \lambda \, \lambda^t \rho(\lambda).
\end{equation}
One of the central mathematical results for both Nevanlinna--Pick
interpolation and the Hamburger moment problem is the presence of
matrix positivity conditions that are both necessary and sufficient
for solutions to exist.
In the case of Nevanlinna-Pick interpolation, the positivity condition
is most naturally formulated for interpolation problems on the unit
disk $\D = \{z \in \C \mid |z| \leq 1\}$,
for which the positivity condition is given by the following theorem:
\begin{theorem}[Pick~\cite{Pick1915}]
  \label{thm:pick-criterion}
Let $z_1, \dots, z_n, w_1, \dots, w_n \in \D$. Then there exists $f :
\D \to \D$ such that $f(z_i) = w_i$ for all $i \in \{1, \dots, n\}$ if
and only if the \emph{Pick matrix}
\begin{equation}
  \label{eq:pick-matrix-def}
  P_{ij} \equiv \frac{1 - w_i w_j^*}{1 - z_i z_j^*}
\end{equation}
is positive-definite.
\end{theorem}
Although this theorem does not directly apply to the physical interpolation problems on $\H$, they can be transformed into an interpolation problem on $\D$ via suitable conformal maps; see Refs.~\cite{Bergamaschi:2023xzx,PhysRevLett.126.056402,Nogaki:2023mut} for details.
Meanwhile, for the Hamburger moment problem, the positivity condition
is stated directly in terms of the correlation function:
\begin{theorem}[Hamburger~\cite{Hamburger:1920}]
Let $C_0, \dots, C_{n-1} \in \R$. Then there exists $\rho \geq 0$ with
moments $\{C_t\}$ if and only if the Hankel matrix
\begin{equation}
  \label{eq:hankel-matrix-def}
  H_{ij} \equiv (C_{i + j})_{i,j=0}^{n-1}
  = \begin{pmatrix}
	C_0 & C_1 & C_2 & \dots \\
	C_1 & C_2 & C_3 & \dots \\
	C_2 & C_3 & C_4 & \dots \\
	\vdots &  \vdots &  \vdots & \ddots
  \end{pmatrix}
\end{equation}
is positive-definite.
\end{theorem}

\section{Geometry of Causal Subsets}\label{sec:geometry}
In the mathematical literature, the interpolation data is typically assumed to have infinite precision with no uncertainty. In contrast, lattice field theory calculations are typically Monte-Carlo calculations, which provide noisy Euclidean data at finite precision with comparatively large uncertainties.
In this context it is important to consider not only whether a given, fixed set of data admits a causual interpolation, but further
how the set of causally-consistent data are distributed throughout the
space of possible Euclidean data. In this direction, suppose that a
set of interpolation nodes $\vec{z} = (z_1, \dots, z_n) \in \D^n$ is
exactly known and fixed,\footnote{This is true the case for
applications of Pick interpolation to lattice field theory since the
interpolation nodes are derived from the Matsubara frequencies, which
are analytically known~\cite{Bergamaschi:2023xzx,PhysRevLett.126.056402}.} and define
the \emph{Pick space} $\P_{\vec{z}}$ to be the set of all $\vec{w} \in
\D^n$ satisfying the Pick criterion, i.e.
\begin{equation}
    \P_{\vec{z}} = \left\{
    (w_1, \dots, w_n) \,\left|\, \frac{1 - w_i w_j^*}{1 - z_i z_j^*} \succeq 0
    \right.
    \right\}.
\end{equation}
Similarly, in the case of the Hamburger moment problem one may define
an analagous space
\begin{equation}
    \P_H = \left\{(C_0, \dots, C_{n-1}) \,\left|\, (C_{i+j})_{i,j=0}^{(n-1)/2} \succeq 0 \right. \right\}.
\end{equation}
A key fact underlying the geometry of both problems is that their
causal spaces are convex. For $\P_H$, convexity is
straightforward to see, since the space of positive-definite Hankel
matrices is a linear subspace of the set of positive-definite
matrices, and the set of positive-definite matrices is convex. 
The convexity of PSD Hankel matrices has recently been exploited elsewhere in the physics literature on spectral functions in Ref.~\cite{Yu_2024}.
For $\P_{\vec{z}}$, the convexity is no longer as apparent due to the fact
that the Pick matrix is not linear in the inputs $\{w_i\}$.
However, with a slight change of perspective, it is straightforward to
prove that $\P_{\vec{z}}$ is convex for any choice of $\vec{z}$,
resulting in the following lemma:
\begin{lemma}
For any fixed set of interpolating nodes $\{z_i\}$, the Pick space
$\P_z$ is convex.\footnote{
To the authors' knowledge, this result appears to be novel in the physics literature. 
However, given the depth of literature on these topics and the simplicity of the proof, it seems likely that such a result has been known previously.
Convexity is alluded to in the mathematical literature, e.g., in Ref.~\cite{Nicolau2016}.
}
\end{lemma}
\begin{proof}
Let $w^{(0)},  w^{(1)} \in \P_z$, and let $\lambda \in [0, 1]$ be
arbitrary. Then it suffices to prove that the convex linear
combination
\begin{equation}
  \label{eq:w-lambda-def}
w^{(\lambda)} = (1 - \lambda) w^{(0)} + \lambda w^{(1)}
\end{equation}
is contained within $\P_z$ as well.
By \cref{thm:pick-criterion} there exist $f_0, f_1 : \D \to \D$ such
that
\begin{equation}
  \label{eq:f-interp-boundary}
  \begin{aligned}
	f_0(z_i) &= w^{(0)}_i \\
	f_1(z_i) &= w^{(1)}_i
  \end{aligned}
\end{equation}
for both $i$. Define a function $f_{\lambda} : \D \to \D$
by\footnote{Note that the image of $f_{\lambda}$ lies within $\D$
since the unit disk is convex.}
\begin{equation}
  \label{eq:f-interp-lambda}
    f_\lambda(z) = (1 - \lambda) f_0(z) + \lambda f_1(z).
\end{equation}
Combining
\cref{eq:f-interp-boundary,eq:w-lambda-def,eq:f-interp-lambda} then
yields the result
\begin{equation}
\label{eq:1}
f_{\lambda}(z_i) = w^{(\lambda)}_i
\end{equation}
for both $i$, i.e. $f_{\lambda}$ interpolates $w^{(\lambda)}$, and
consequently $w^{(\lambda)} \in \P_z$ by \cref{thm:pick-criterion} as
expected.
\end{proof}
The proof here is quite general, and applies to Nevanlinna-Pick
interpolation on any domain, as long as the target domain is
convex. In particular, this implies that the original interpolation
problem on $\H$ has the same convexity properties, and also provides
an alternative proof that $\P_H$ is convex provided that the Hamburger
moment problem is reinterpreted as an interpolation problem at infinity (see
Ref.~\cite{Abbott:2025snz} as well as the original mathematical literature \cite{Kovalishina1984}).
Another immediate consequence of convexity is that the boundary of the Pick space is exactly the extremal interpolation problems, i.e. problems for which the Pick
matrix is singular and the interpolating function is unique.

\section{Outlook}\label{sec:conclusion}
Causality and analyticity provide powerful
constraints both on Euclidean-time correlation functions and the
real-time information derived from them.
Nevanlinna--Pick interpolation along with moment problems provide
direct methods of accessing these constraints in the case of exact data.
Given the promising results in Ref.~\cite{Fields:2025glg} in applying NP interpolation to uncertain data, as well success of the closely-related Lanczos method applied to
noisy correlation functions (see Refs.~\cite{Hackett:2024xnx,Hackett:2024nbe,Wagman:2024rid,Abbott:2025yhm}),
it seems plausible that the same rigorous bounds will soon be achievable for practical lattice field theory calculations.
More work is still needed, but extensions of these methods 
may eventually be able to be used to construct bounds incorporating the full information content of lattice correlators. 
An important question then becomes which set of correlators can be used to establish the strongest bounds.
In that direction, one aspect of moment problems worth emphasizing
is the ability to use matrix correlators to constrain spectral
functions (see Ref.~\cite{Abbott:2025snz}).
Given the widespread use of matrix correlators for extracting spectral
information across many systems (e.g. via a generalized eigenvalue problem), using multiple
operators for extracting spectral functions will likely be crucial in
reducing uncertainties and increasing precision of future
calculations, as well as enabling new calculations that otherwise would not have been possible.

\section*{Acknowledgements}
We thank Norman Christ, Sarah Fields, and Matteo Saccardi
for useful discussions and comments on the manuscript.
RA is an Ernest Kempton Adams Postdoctoral Fellow supported in part by the Ernest Kempton Adams fund for Physical Research of Columbia University
and in part by U.S. DOE grant No. DE-SC0011941
PO is supported in part by the U.S. Department of Energy, Office of Science, Office of Nuclear Physics under grant Contract Number DE-SC0012704 (BNL).

\bibliographystyle{JHEP}
\bibliography{main.bib}

\end{document}